\documentclass{article}
\usepackage{graphicx}
\usepackage{amsmath,amssymb,amsthm}
\usepackage{algorithm}
\usepackage{algorithmic}
\usepackage{dsfont}
\usepackage[T1]{fontenc}
\usepackage{color}

\usepackage{authblk}
\usepackage{geometry}
\newtheorem{theorem}{Theorem}
\newtheorem{lemma}{Lemma}
\newtheorem{observation}{Observation}
\theoremstyle{definition}
\newtheorem{definition}{Definition}

\newcommand{\ie}{{\it i.e., }} 
\newcommand{\etal}{{\it et al.}} 
 
\newcommand{\eg} {{\it e.g., }} 

\newcommand{\suchthat}{:}
\newcommand{\ex}[1]{\mathbf{E}\left[ #1 \right]}
\newcommand{\relmiddle}[1]{\mathrel{}\middle#1\mathrel{}}
\newcommand{\longmid}{\relmiddle{|}}

\newcommand{\ifthen}[2]{\textbf{if} #1 \textbf{then} #2 \textbf{endif}}
\newcommand{\foralldo}[2]{\textbf{for all} #1 \textbf{do} #2 \textbf{endfor}}
\newcommand{\vblank}{\vspace{0.2cm}}

\newcommand{\elect}{\mathbf{LE}}
\newcommand{\ranking}{\mathbf{RK}}
\newcommand{\degree}{\mathbf{DR}}
\newcommand{\neighbor}{\mathbf{NR}}

\newcommand{\prank}{P_{\mathrm{rank}}}
\newcommand{\pneighbor}{P_{\mathrm{neigh}}}

\newcommand{\call}{\mathcal{C}_{\mathrm{all}}}
\newcommand{\safe}{\mathcal{S}}
\newcommand{\stoken}{\mathcal{S}_{\mathrm{token}}}
\newcommand{\ssync}{\mathcal{S}_{\mathrm{sync}}}
\newcommand{\srank}{\mathcal{S}_{\mathrm{rank}}}
\newcommand{\rs}{\mathbf{\Gamma}} 
\newcommand{\snofake}{\mathcal{S}_{\mathrm{noFake}}}
\newcommand{\calc}{\mathcal{C}}

\newcommand{\id}{\mathtt{id}}
\newcommand{\idA}{\mathtt{id}_A}
\newcommand{\idT}{\mathtt{id}_T}
\newcommand{\colorA}{\mathtt{color}_A}
\newcommand{\colorT}{\mathtt{color}_T}
\newcommand{\degreeT}{\mathtt{degree}_T}
\newcommand{\dsum}{\mathtt{sum}}
\newcommand{\counted}{\mathtt{counted}}

\newcommand{\resetE}{\mathtt{reset}_E}

\newcommand{\timerP}{\mathtt{timer}_P}

\newcommand{\timerT}{\mathtt{timer}_T}

\newcommand{\neighbors}{\mathtt{neighbors}}

\newcommand{\tmax}{U_T}
\newcommand{\emax}{U_E}
\newcommand{\pmax}{U_P}

\newcommand{\swap}{\leftrightarrow}

\newcommand{\shiro}{W}
\newcommand{\aka}{R}
\newcommand{\ao}{B}

\newcommand{\fl}{\mathit{false}}

\newcommand{\graph}{\mathcal{G}}
\newcommand{\outputs}{\pi_{\mathrm{out}}}

\newcommand{\ndomain}{\mathbb{N}_{\ge 2}}
\newcommand{\mdomain}{\mathbb{N}_{\ge 1}}

\newcommand{\qtoken}{Q_{\mathrm{token}}}

\newcommand{\nlabel}{L_{\mathrm{neigh}}}
\newcommand{\chainX}{\mathbf{X}}
\newcommand{\chainY}{\mathbf{Y}}
\newcommand{\chainZ}{\mathbf{Z}}
\newcommand{\piX}{\pi_{\chainX}}
\newcommand{\piY}{\pi_{\chainY}}
\newcommand{\piZ}{\pi_{\chainZ}}

\newcommand{\states}{S}

\begin{document}
\title{
The Power of Global Knowledge on Self-stabilizing Population Protocols
\thanks{
This work was supported by JSPS KAKENHI Grant Numbers 17K19977, 18K18000, 18K18029, 18K18031, 19H04085, and 20H04140 and JST SICORP Grant Number JPMJSC1606.
}
}

\author[1]{Yuichi Sudo\thanks{Corresponding author:y-sudou[at]ist.osaka-u.ac.jp}}
\author[2]{Masahiro Shibata}
\author[3]{Junya Nakamura}
\author[4]{Yonghwan Kim}
\author[1]{Toshimitsu Masuzawa}

\affil[1]{Osaka University, Japan}
\affil[2]{Kyushu Institute of Technology, Japan}
\affil[3]{Toyohashi University of Technology, Japan}
\affil[4]{Nagoya Institute of Technology, Japan}

\date{}

\maketitle              

\begin{abstract}
In the population protocol model,
many problems cannot be solved in a self-stabilizing way.
However, global knowledge, such as the number of nodes in a network,
sometimes allow us to design a self-stabilizing protocol for such problems.
In this paper, we investigate
the effect of global knowledge on
the possibility of self-stabilizing population protocols
in arbitrary graphs.
Specifically, we clarify the solvability
of the leader election problem, the ranking problem,
the degree recognition problem, and
the neighbor recognition problem
by self-stabilizing population protocols
with knowledge of the number of nodes
and/or the number of edges in a network.
\end{abstract}

\section{Introduction}
\label{sec:intro}
We consider the \emph{population protocol} (PP) model \cite{original} in
 this paper.
A network called \emph{population} consists of a large number of
 finite-state automata,
called \emph{agents}.
Agents make \emph{interactions}
 (\ie pairwise communication) with each other
by which they update their states.
The interactions are opportunistic,
that is, they are unpredictable for the agents.
Agents are strongly anonymous:
they do not have identifiers
and they cannot distinguish
their neighbors with the same states.
One example represented by this model is
a flock of birds where each bird is equipped with
a sensing device with a small transmission range.
Two devices can communicate (\ie interact) with each other
only when the corresponding birds
come sufficiently close to each other.
Therefore, 
an agent cannot predict when it has its next interaction.

In the field of population protocols, many efforts have been devoted to devising protocols for a complete graph, that is,
a population where every pair of agents interacts
infinitely often.
On the other hand,
several works
\cite{original,jikoantei,twooracles,onebit,CC19,CG17,MNRS14,samespeed,ninikanjiko,Sud+20B}
study 
the population represented by a general graph $G=(V,E)$
where $V$ is the set of agents and
$E$ specifies the set of interactable pairs.
Each pair of agents $(u,v)\in E$ has interactions
infinitely often,
while each pair of agents $(u',v') \notin E$ never has an interaction.


\emph{Self-stabilization} \cite{ganso}
is a fault-tolerant property that,
even when any transient fault
(\eg memory crash) hits a network,
it can autonomously recover from the fault.
Formally, self-stabilization is defined as follows:
(i) starting from an arbitrary configuration,
a network eventually reaches a \emph{safe configuration} (\emph{convergence}),
and
(ii) once a network reaches a safe configuration,
it keeps its specification forever (\emph{closure}).
Self-stabilization is of great importance in the PP model
because self-stabilization tolerates any finite number of transient faults, and this is a necessary property
in a network consisting of a huge number of cheap and unreliable nodes. 

Consequently,
many studies have been devoted to self-stabilizing population protocols \cite{jikoantei,twooracles,SIW12,onebit,CC19,oracle,kakai,samespeed,kanjiko,ninikanjiko,Sud+20B,Sud+20}. Angluin \etal~\cite{jikoantei} gave self-stabilizing protocols for a variety of problems: the leader election in the rings whose size are not multiples of a given integer $k$ (in particular, the rings of odd size), the token circulation in rings with a pre-selected leader, the 2-hop coloring in degree-bounded graphs, the consistent global orientation in undirected rings, and the spanning-tree construction in regular graphs. The protocols for the first four problems use only a constant space of agent memory,
while the protocol for the last problem requires $O(\log D)$ bits of agent memory, where $D$ is (a known upper bound\footnote{In \cite{jikoantei}, $D$ is defined as the diameter of the graph, not a known upper bound on it. However, since we must take into account an arbitrary initial configuration, we require an upper bound on the diameter; Otherwise, the agents need the memory of unbounded size. Fortunately, the knowledge of the upper bound is not a strong assumption in this case: any upper bound which is polynomial in the true diameter is acceptable since the space complexity is $O(\log D)$ bits.}
on) the diameter of the graph.
Chen and Chen \cite{CC19} gave a constant-space and self-stabilizing
protocol for the leader election
in rings with arbitrary size.

On the negative side,
Angluin \etal~\cite{jikoantei} proved that the self-stabilizing leader election (SS-LE) is impossible for arbitrary graphs. In particular, it immediately follows from their theorem that no protocol solves SS-LE in complete graphs with three different sizes, \ie in all of $K_i$, $K_j$, and $K_k$ for any distinct integers $i,j,k \ge 2$, where $K_l$ is a complete graph with size $l$. Cai \etal~\cite{SIW12} proved that no protocol solves SS-LE both in $K_i$ and in $K_{i+1}$ for any integer $i \ge 2$. In almost the same way, we can easily observe that
no protocol solves SS-LE both in $K_i$ and $K_j$ for any distinct integers $i,j \ge 2$. (See a more detailed explanation in the second page of \cite{Sud+20}.)
In other words, SS-LE is impossible unless the exact number of agents in the population is known to the agents.  Because Cai \etal~\cite{SIW12} also gave a protocol that solves SS-LE in $K_l$ for a given integer $l$, the knowledge of the exact number of agents is necessary and sufficient to solve SS-LE in a complete graph.

In addition to \cite{jikoantei,SIW12,CC19}, many works have been devoted to SS-LE. This is because
the leader election is one of the most fundamental and important problems in the PP model:
several important protocols \cite{original,fast,jikoantei} require a pre-selected unique leader,
especially, it is shown by Angluin \etal~\cite{fast} that if we have a unique leader,
all semi-linear predicates can be solved very quickly.
However, we have strong impossibility as mentioned above:
SS-LE can not be solved unless the knowledge of
the exact number of agents is given to the agents.
In the literature, there are three approaches
to overcome this impossibility.
One approach \cite{Bur+19,SIW12} is to assume that every agent knows the exact number of agents.
Cai \etal~\cite{SIW12} took this approach for the first time.
Their protocol uses $O(\log n)$ bits ($n$ states) of memory space per agent
and converges within $O(n^3)$ steps in expectation in the complete graph of $n$ agents
under \emph{the uniformly random scheduler},
which selects a pair of agents to interact uniformly at random from all pairs
at each step.
Burman \etal~\cite{Bur+19} gave three \emph{faster} SS-LE protocols
than the protocol of Cai \etal~\cite{SIW12}, also for the complete graph of $n$ agents.
These self-stabilizing protocols in \cite{Bur+19,SIW12} solve not only the leader election problem
but also \emph{the ranking problem},
which requires ranking the $n$ agents
by assigning them the different integers from $0,1,\dots,n-1$.
See Section \ref{sec:related}
for the results of the other two approaches
to overcome the impossibility,
SS-LE protocols with \emph{oracles} \cite{twooracles,onebit,oracle} and \emph{loosely-stabilizing protocols} \cite{kakai,samespeed,kanjiko,ninikanjiko,Sud+20B,Sud+20}.

\subsection{Our Contribution}
\label{sec:contribution}
As mentioned above, if we have
knowledge of the exact number of agents,
we can solve the self-stabilizing leader election
in complete graphs,
which we can never solve otherwise.
In this paper, we investigate in detail how powerful
global knowledge, such as the exact number of agents in the population,
is to design self-stabilizing population protocols
for \emph{arbitrary graphs}. 
Specifically, we consider two kinds of global knowledge,
the number of agents and the number of edges (\ie interactable pairs)
in the population, and clarify the relationships between the knowledge
and the solvability of the following four problems:
\begin{itemize}
 \item leader election ($\elect$): Elect exactly one leader,
 \item ranking ($\ranking$): Assign the agents in the population $G=(V_G,E_G)$ distinct integers (or \emph{ranks})
from $0$ to $|V_G|-1$,
  \item degree recognition ($\degree$): Let each agent recognize its degree in the graph,
 \item neighbor recognition ($\neighbor$): Let each agent recognize the set of its neighbors in the graph.
Since the population is anonymous, this problem
also requires having 2-hop coloring,
that is, all agents must be assigned integers (or \emph{colors})
such that all neighbors of any agent have different colors.
\end{itemize}
In addition to the above specifications,
we require that no agent change its outputs
(e.g., its \emph{rank} in $\ranking$)
after the population converges, that is,
it reaches a safe configuration.

We denote $A_1 \preceq A_2$ if
problem
$A_1$ is reducible to $A_2$.
We have 
$\elect \preceq \ranking$ and 
$\degree \preceq \neighbor$.
The first relationship holds
because if the agents are labeled $0,1,\dots,|V_G|-1$,
$\elect$ is immediately solved
by selecting the agent with label 0
as the unique leader.
The second relationship is trivial.

To describe our contributions, 
we formally define the global knowledge that we consider.
Define $\graph_{n,m}$ as the set of all the simple, undirected,
and connected graphs with $n$ nodes and $m$ edges.
Let $\nu$ and $\mu$ be any sets of positive integers
such that $\nu \subseteq \ndomain = \{n \in \mathbb{N} \mid n \ge 2\}$
and $\mu \subseteq \mdomain = \{m \in \mathbb{N} \mid m \ge 1\}$.
Then, we define
$\graph_{\nu,\mu}=\bigcup_{n \in \nu, m \in \mu} \graph_{n,m}$.
For simplicity, we define $\graph_{\nu,*}=\graph_{\nu,\mdomain}$ and $\graph_{*,\mu}=\graph_{\ndomain,\mu}$
for any $\nu \subseteq \ndomain$ and $\mu \subseteq \mdomain$.
We consider that $\nu$ and $\mu$ are global knowledge on the population: $\nu$ is the set of the possible numbers of agents
and $\mu$ is the set of the possible numbers of interactable pairs. In other words, when we are given $\nu$ and $\mu$,
our protocol has to solve a problem only in the populations
represented by the graphs in $\graph_{\nu,\mu}$.
We say that protocol $P$ solves problem $A$
\emph{in arbitrary graphs}
given knowledge $\nu$ and $\mu$
if $P$ solves $A$ in all graphs in $\graph_{\nu,\mu}$.

In this paper,
we investigate the solvability
of 
$\elect$, $\ranking$, $\degree$, and $\neighbor$
for arbitrary graphs
with the knowledge $\nu$ and $\mu$.
Specifically, we prove the following propositions
assuming that the agents
are given knowledge $\nu$ and $\mu$:
\begin{enumerate}
\item
When the agents know nothing about the number of interactable pairs, \ie $\mu=\mdomain$,
there exists a self-stabilizing protocol
that solves $\elect$ and $\ranking$ in arbitrary graphs
if and only if
the agents know the exact number of agents
\ie $\graph_{\nu,\mu}=\graph_{n,*}$
for some $n \in \ndomain$.     
\item
There exists a self-stabilizing protocol
that solves $\neighbor$ ($\succeq \degree$) in arbitrary graphs
if  
the agents know the exact number of agents
and the exact number of interactable pairs
\ie 
$\graph_{\nu,\mu} = \graph_{n,m}$ holds
for some $n \in \ndomain$ and $m \in \mdomain$.
\item
The knowledge of the exact number of agents is not enough
to design a self-stabilizing protocol that solves $\degree$ ($\preceq \neighbor$)
in arbitrary graphs
if the agents do not know the number of interactable pairs \emph{exactly}.
Specifically,
no self-stabilizing protocol 
solves $\degree$ in all graphs in $\graph_{\nu,\mu}$
if $\graph_{n,m_1} \cup \graph_{n,m_2} \subseteq \graph_{\nu,\mu}$
holds for some $n \in \ndomain$
and some distinct $m_1,m_2 \in \mdomain$
such that $\graph_{n,m_1} \neq \emptyset$
and $\graph_{n,m_2} \neq \emptyset$.
\end{enumerate}

In standard distributed computing models, generally,
each node always has its local knowledge,
\eg its degree and the set of its neighbors.
In the PP model,
the agents does not have the local knowledge \emph{a priori},
and many impossibility results 
(\eg the impossibility of SS-LE
in complete graphs \cite{jikoantei,SIW12})
come from 
the lack of the local knowledge.
Interestingly, 
the third proposition yields that,
for self-stabilizing population protocols,
obtaining some local knowledge (degree recognition of each agent)
is at least as difficult as obtaining the corresponding
global knowledge (the number of interactable pairs).
It is also worthwhile to mention that
the PP model is empowered greatly
if $\elect$ and $\neighbor$ are solved.
After 
the agents recognize their neighbors correctly,
the population can simulate one of the most standard
distributed computing models, the message passing model,
if each agent maintains a variable corresponding to a \emph{message buffer} for each neighbor.
Moreover, we have the unique leader in the population,
by which we can easily break the symmetry of a graph
and solve many important problems even in a self-stabilizing way.
For example, we can construct a spanning tree rooted by the leader.
This fact and the above propositions show how powerful
this kind of global knowledge is
when we design self-stabilizing population protocols.



\subsection{Other Related Work}
\label{sec:related}
Several works use \emph{oracles},
a kind of failure detectors,
to solve SS-LE.
Fischer and Jiang \cite{oracle}
took this approach for the first time.
They introduced oracle $\Omega ?$ that eventually
tells all agents whether at least one leader exists or not
and proposed two protocols that solve SS-LE for rings and complete
graphs by using $\Omega ?$.
Beauquier \etal~\cite{twooracles}
presented an SS-LE protocol
for arbitrary graphs that uses two copies of $\Omega ?$,
one is used to detect the existence of a leader
and the other one is used to detect the existence
of a special agent called \emph{a token}.
Canepa \etal~\cite{onebit} proposed two SS-LE protocols that use $\Omega ?$ and require
only 1 bit of each agent:
one is a deterministic protocol for trees
and the other is a randomized protocol for arbitrary graphs
although the position of the leader is not static 
and moves among the agents forever.

To solve SS-LE without oracles or the knowledge of the exact number of agents, 
Sudo \etal~\cite{kanjiko} introduced
the concept of 
$\emph{loose-stabilization}$, 
which relaxes the closure requirement of 
self-stabilization, 
but keeps its advantage in practice. 
Specifically, loose-stabilization guarantees 
that, 
starting from any configuration, 
the population reaches a safe 
configuration 
within a relatively short time; 
after that, the specification of the problem 
(such as having a unique leader) 
must be sustained for a sufficiently long 
time, though not necessarily forever. 
In \cite{kanjiko}, 
a loosely-stabilizing leader election (LS-LE) 
protocol
was given for the first time, which assumes that
the population is a complete graph and
every agent knows a common 
\emph{upper bound} 
$N$ of $n$, where $n$ is the number of agents in the population. 
This protocol is practically equivalent 
to an SS-LE protocol since 
it maintains the unique leader for 
an exponentially large number of steps in expectation
(that is, practically forever) 
after reaching a safe configuration within 
$O(nN \log N)$ steps in expectation.
The assumption that we can use an upper bound $N$ of $n$ 
is practical because the protocol works 
correctly 
even if we make a large overestimation of $n$, 
such as $N = 10n$. 
Izumi \cite{kakai} gave a method 
which reduces the number of steps for convergence 
to $O(nN)$. 
Later, Sudo \etal~\cite{Sud+20}
gave a much faster loosely-stabilizing leader election protocol
for complete graphs.
Given parameter $\tau \ge 10$, 
it reaches a safe configuration within $O(\tau n \log^3 N)$ steps
and thereafter it keeps the unique leader
for $\Omega(cn^{10\tau})$ steps, both in expectation.
Very recently, Sudo \etal~\cite{Sud+20d}
gave a time optimal protocol for complete graphs:
the convergence time is $O(\tau n \log N)$ steps
and the holding steps is $\Omega(n^{\tau+1})$ steps.
In \cite{samespeed,ninikanjiko,Sud+20B}, 
LS-LE protocols 
were presented for arbitrary graphs.

\section{Preliminaries}
\label{sec:pre}
A \emph{population} is 
represented by a simple and connected graph $G=(V_G,E_G)$,
where $V_G$ is the set of the \emph{agents}
and $E_G \subseteq V_G \times V_G$ is the set of the \emph{interactable pairs} of agents.
If $(u,v) \in E_G$, 
two agents $u$ and $v$ can interact in the population $G$,
where $u$ serves as the \emph{initiator} 
and $v$ serves as the \emph{responder} of the 
interaction. 
In this paper, we consider only \emph{undirected} populations,
that is, we assume that, for any population $G$, 
$(u,v) \in E_G$ yields $(v,u) \in E_G$
for any $u,v \in V_G$.
We define
the set of the neighbors of agent $v$
as $N_G(v)=\{u \in V_G \mid (v,u) \in E_G\}$.

%

A \emph{protocol} $P(Q,Y,T,\outputs)$ 
consists of 
a finite set $Q$ of states, 
a finite set $Y$ of output symbols, 
a transition function 
$T: Q \times Q \to Q \times Q$, 
and an output function $\outputs : Q \to Y$. 
When two agents interact, 
$T$ determines their next states 
according to their current states. 
The \emph{output} of an agent is determined 
by $\outputs$: 
the output of an agent in state $q$ is 
$\outputs(q)$. 
As mentioned in Section \ref{sec:intro},
we assume that the agents can use knowledge
$\nu$ and $\mu$.
Therefore, the four parameters of protocol $P$,
\ie $Q$, $Y$, $T$, and $\outputs$,
may depend on $\nu$ and $\mu$.
We sometimes write $P(\nu,\mu)$ explicitly
to denote protocol $P$ with knowledge $\nu$ and $\mu$.

A \emph{configuration} on population $G$
is a mapping
$C : V_G \to Q$ that specifies 
the states of all the agents in $G$. 
We denote the set of all configurations of 
protocol $P$ on population $G$
by $\call(P,G)$. 
We say that a configuration $C$ changes to 
$C'$ by an interaction 
$e = (u,v)$, 
denoted by $C \stackrel{P,e}{\to} C'$, 
if 
$(C'(u),C'(v))=T(C(u),C(v))$ 
and $C'(w) = C(w)$ 
for all $w \in V \setminus \{u,v\}$. 
We also denote $C \stackrel{P,G}{\to} C'$
if $C \stackrel{P,e}{\to} C'$ holds for some $e \in E_G$.
We also say that a configuration $C'$ is reachable from $C$
by $P$ on population $G$ if there is a sequence of configurations
$C_0,C_1,\dots,C_k$ such that $C_i \stackrel{P,G}{\to} C_{i+1}$
for $i=0,1,\dots,k-1$.
We say that a set $\safe$ of configurations
is \emph{closed}
if no configuration out of $\safe$
is reachable from a configuration in $\safe$.

An \emph{execution} of protocol $P$ on population $G$
is an infinite sequence of configurations $\Xi=C_0,C_1,\dots$
such that $C_i \stackrel{P,G}{\to} C_{i+1}$ for $i=0,1,\dots$.
We call $C_0$ the initial configuration of the execution $\Xi$.
We have to assume some kind of fairness of an execution.
Otherwise, for example, we cannot exclude an execution
such that only one pair of agents have interactions in a row
and no other pair has an interaction forever.
Unlike most distributed computing models in the literature, 
\emph{the global fairness} is usually assumed in the PP model.
We say that an execution $\Xi=C_0,C_1,\dots$ of $P$
on population $G$ satisfies
the global fairness (or $\Xi$ is globally fair)
if
for any configuration $C$ that appears infinitely often in $\Xi$,
every configuration $C'$ such that $C \stackrel{P,G}{\to} C'$
also appears infinitely often in $\Xi$.

A problem is specified by a predicate on the outputs of the agents.
We call this predicate the \emph{specification} of the problem.
We say that a configuration $C$ satisfies the specification
of a problem if the outputs of the agents satisfy it
in $C$.
We consider the following four problems in this paper.
\begin{definition}[$\elect$]
The specification of the leader election problem ($\elect$)
requires that exactly one agent outputs $L$
and all the other agents output $F$.
\end{definition}
\begin{definition}[$\ranking$]
The specification of the ranking problem ($\ranking$)
requires that in the population $G=(V_G,E_G)$,
the set of the outputs of the agents in the population
equals to $\{0,1,\dots,|V_G|-1\}$.
\end{definition}
\begin{definition}[$\degree$]
The specification of the degree recognition problem ($\degree$)
requires that in the population $G=(V_G,E_G)$, 
every agent $v \in V_G$ outputs $|N_G(v)|$.
\end{definition}
\begin{definition}[$\neighbor$]
The specification of the neighbor recognition problem ($\neighbor$)
requires that in the population $G=(V_G,E_G)$, 
every agent $v \in V_G$ outputs a two-tuple
$(c_v,S_v) \in \mathbb{Z} \times 2^{\mathbb{Z}}$
such that, for all $v \in V_G$,
we have 
$S_v=\{c_u \mid u \in N_G(v)\}$
and $|S_v|=|N_G(v)|$.
\end{definition}
Note that the second condition in the definition of $\neighbor$,
\ie $|S_v|=|N_G(v)|$, requires that
the population is 2-hop colored,
that is, every two distinct neighbors $u$ and $w$ of agent $v$ must have different integers $c_u$ and $c_w$.

Now, we define self-stabilizing protocols in Definitions \ref{def:safe} and \ref{def:ss},
where we use the definitions
given in Section \ref{sec:contribution}
for knowledge $\nu$ and $\mu$ 
and the set $\graph_{\nu,\mu}$ of graphs.
Note that Definition \ref{def:safe} is not enough if we consider
\emph{dynamic problems} such as the token circulation,
where the specifications must be defined as predicates
not on configurations but on executions. 
However, we consider only static problems in this paper,
thus this definition is enough for our purpose.

\begin{definition}[Safe configuration]
\label{def:safe} 
Given a protocol $P$ and a population $G$,
we say that a configuration $C \in \call(P(\nu,\mu),G)$
is \emph{safe} for problem $A$
if (i) $C$ satisfies the specification of problem $A$,
and (ii) no agent changes its output in any execution of $P$ on $G$
starting from $C$.
\end{definition}

\begin{definition}[Self-stabilizing protocol]
\label{def:ss}
For any $\nu$ and $\mu$,
we say that a protocol $P$ is a self-stabilizing protocol
that solves problem $A$ in arbitrary graphs
given knowledge $\nu$ and $\mu$
if
every globally-fair execution
of $P(\nu,\mu)$
on any population $G$,
which starts from any configuration $C_0 \in \call(P(\nu,\mu),G)$,
reaches a safe configuration for $A$.
\end{definition}

Finally, we define \emph{the uniformly random scheduler},
which has been considered in most of the works 
\cite{AG15,original,fast,GSU18,samespeed,kanjiko,SOI+19+,ninikanjiko,Sud+20B,Sud+20}
in the PP model.
Under this scheduler,
exactly one ordered pair $(u,v) \in E_G$ is chosen to interact
uniformly at random from all interactable pairs.
We need this scheduler to evaluate
\emph{time complexities} of protocols
because
global fairness only guarantees that
an execution makes progress \emph{eventually}.
Formally, the uniformly random scheduler is defined as
a sequence of interactions $\rs=\Gamma_0, \Gamma_1,\dots$,
where each $\Gamma_t$ is a random variable such that 
$\Pr(\Gamma_t = (u,v)) =1/|E_G|$ 
for any $t \ge 0$ and any $(u,v) \in E_G$. 
Given a population $G$, a protocol $P(\nu,\mu)$, and 
an initial configuration $C_0 \in \call(P(\nu,\mu),G)$,
the execution under the uniformly random scheduler
is defined as 
$\Xi_{P(\nu,\mu)}(G,C_0,\rs) = C_0,C_1,\dots$ such 
that 
$C_t \stackrel{P(\nu,\mu),\Gamma_t}{\to} C_{t+1}$ for 
all $t \ge 0$. 
When we assume this scheduler,
we can evaluate time complexities of a population protocol,
for example, the expected number of steps
required to reach a safe configuration.
We have the following observation.
\begin{observation}
\label{obs:ss}
A protocol $P(\nu,\mu)$ is self-stabilizing
for a problem $A$
if and only if
$\Xi_{P(\nu,\mu)}(G,C_0,\rs)$ reaches
a safe configuration for $A$ with probability $1$
for any configuration $C_0 \in \call(P(\nu,\mu),G)$.
\end{observation}
\begin{proof}
Remember that we do not allow
a protocol to have an infinite number of states.
According to \cite{original},
we say that a set $\calc$ of configurations is \emph{final}
if $\calc$ is closed,
and all configurations in $\calc$ are reachable from 
each other.
We also say that
a configuration $C$ is \emph{final}
if it belongs to a final set.
It is trivial that
protocol $P$ is self-stabilizing
if and only if all final configurations are safe. 
Thus, 
it suffices to show that
execution $\Xi=\Xi_{P(\nu,\mu)}(G,C_0,\rs)$ reaches
a safe configuration for $A$ with probability $1$
for any $C_0 \in \call(P(\nu,\mu),G)$
if and only if
all final configurations of $P(\nu,\mu)$ are safe for $A$.
The sufficient condition is trivial
because $\Xi$ reaches a final configuration with probability 1
regardless of $C_0$. We prove the necessary condition below.
Suppose that there is a final configuration $C$ that is not safe. 
By definition, $C$ belongs to
a final set $\calc$.
Since $C$ is reachable from all configurations in $\calc$,
no configuration in $\calc$ is safe.
Since $\calc$ is closed, $\Xi$ will never reach a safe configuration if $C_0=C$.
\end{proof}



\section{Random Walk in Population Protocols}
\label{sec:randomwalk}
In this paper, we gave two self-stabilizing protocols 
$\prank$ and $\pneighbor$.
Both of them use $n$ tokens that make the random walk,
where $n$ is the number of agents in the population.
Specifically, all the agents in the population
always has exactly one token,
and two agents swap their tokens
whenever they have an interaction.
In this section,
we give several lemmas about the movements of the tokens
(\ie Lemmas \ref{lem:hittingtime}, \ref{lem:covertime},
\ref{lem:meet_all_tokens}, \ref{lem:meet_each_other},
\ref{lem:some_drift}, and \ref{lem:drift})
to analyze the expected number of steps until
an execution of $\prank$ or $\pneighbor$
reaches a safe configuration.

Although Lemmas \ref{lem:hittingtime} and
\ref{lem:meet_each_other} were already proven
by Sudo \etal~\cite{Sud+20B},
we also give proofs for them
with the notations of this paper,
to make this paper self-contained.

Fix a population $G=(V_G, E_G)$
and consider
the execution $\Xi=C_0,C_1,\dots$
of $\prank$ or $\pneighbor$\footnotemark{} 
under the uniformly random scheduler
starting from an arbitrary configuration $C_0$.
\footnotetext{
In this section,
we do not care which protocol, $\prank$ or $\pneighbor$,
we execute
because we focus on only the movement of the tokens
making the random walk.
}
Let $\rs=\Gamma_0, \Gamma_1,\dots=(u_0,v_0),(u_1,v_1),\dots$,
that is, we denote the $i$-th interaction
under the uniformly random scheduler $\rs$ by $(u_i,v_i)$.
Formally, for each $w \in V_G$
we define token $t_w: \mathbb{N}_{\ge 0} \to V_G$ as follows:
\begin{itemize}
 \item $t_w(0)=w$,
 \item 
$t_w(i)=
\begin{cases}
u_t & t_w(i-1) = v_{i-1}\\
v_t & t_w(i-1) = u_{i-1}\\
t_w(i-1) & \text{Otherwise}
\end{cases}$\\
for each $i > 0$.
\end{itemize}
We say that token $t_v$ visits $u$ in the $i$-th step
if $t_v(i) = u$.
We also say that two tokens $t_u$ and $t_w$ meet in the $i$-th step if $\Gamma_i = (t_u(i),t_v(i))$ or $\Gamma_i = (t_v(i),t_u(i))$
holds.
In the rest of this section,
we denote the number of agents
and the number of interactable pairs
by $n$ and $m$, that is,
$n = |V_G|$ and $m=|E_G|/2$.
The diameter of population $G$
is denoted by $d$.

\begin{lemma}[\cite{Sud+20B}]
\label{lem:hittingtime}
In execution $\Xi$,
for any $u,v \in V_G$, 
token $t_u$ visits agent $v$
within $mn\cdot d(u,v)$ steps in expectation,
where $d(u,v)$ is the distance between agent $u$ and $v$ in $G$.
\end{lemma}
\begin{proof}
We consider a Markov chain $\chainX = \{X(t) \mid t=0,1,\dots\}$,
where each $X(t) \in V_G$
 represents the location of a token
 (\ie the agents that the token stays on) 
in configuration $C_t$ (\ie the $t$-th configuration in $\Xi$).
For $t \ge 1$ and $x,y \in V_G$,
the probability
$\Pr(X(t) = y \mid X(t-1) = x)$ is independent of $t$,
denoted by $P_{\chainX}(x,y)$.
Probability $P_{\chainX}(x,y)$ is calculated as follows:
$P_{\chainX}(x,y)=1/m$ if $(x,y) \in E_G$,
$P_{\chainX}(x,y)=1-\delta_x/m$ if $x=y$;
otherwise, $P_{\chainX}(x,y)= 0$,
where $\delta_x = |N_G(v)|$.
The symmetric structure of the chain, \ie $P_{\chainX}(x,y)=P_{\chainX}(y,x)$
for all $x,y \in V$, gives $\sum_{x \in V}P_{\chainX}(x,y)=1$
for any $y$.
Therefore, $\piX=(\piX(x_1),\piX(x_2),\dots,\piX(x_n))=(1/n,1/n,\dots,1/n)$ is the unique stationary distribution on $\chainX$
(\ie $\piX P_{\chainX} = \piX$), where $V_G=\{x_1,x_2,\dots,x_{n}\}$.
For $x,y \in V_G$,
we define \emph{the hitting time} $H_{\chainX}(x,y)$
as the expected number of transition steps in the chain
$\chainX$ from state $x$ to $y$.
We have $H_{\chainX}(z,z) = 1/\pi(z) = n$
for any agent $z \in V$.
We also have
$H_{\chainX}(z,z)=1 + \sum_{w \in N_G(z)} (1/m) \cdot H_{\chainX}(w,z)$.
Therefore,
$\sum_{w \in N_G(z)} H_{\chainX}(w,z) = m(n-1)$.
Thus, for any $(w,z) \in E_{G}$,
we have $H_{\chainX}(w,z) < mn$.

Let $w_0,w_1,\dots,w_{d(u,v)}$ ($w_0 = u$ and $w_{d(u,v)} = v$)
be the shortest path
from $u$ to $v$ in $G$.
Then, $H_{\chainX}(u,v) \le \sum_{i = 0}^{l-1} h_{w_i,w_{i+1}} < mn\cdot d(u,v)$, from which the lemma immediately follows.
\end{proof}

\begin{lemma}
\label{lem:covertime}
In execution $\Xi$,
for any $v \in V_G$, 
token $t_v$ visits all the agents in $V_G$
within $2mn^2$ steps in expectation.
\end{lemma}
\begin{proof}
Let $v_0,v_1,\dots,v_{2n-2}$ ($v_0=v_{2n-2}=v$ and $(v_i,v_{i+1}) \in E_G$ for all $i=0,1,\dots,2n-1$) be a tour on an arbitrary spanning tree of $G$.
The lemma immediately follows from Lemma \ref{lem:hittingtime}
because token $t_v$ moves from $v_i$ to $v_{i+1}$ within
$mn$ steps in expectation for each $i=0,1,\dots,2n-1$.
\end{proof}

\begin{lemma}
\label{lem:meet_all_tokens}
In execution $\Xi$,
for any $v \in V_G$,
all the $n$ tokens visit agent $v$
within $O(mnd \log n)$ steps in expectation.
\end{lemma}
\begin{proof}
Let $u$ be any agent in $V_G$.
It immediately follows from Lemma \ref{lem:hittingtime}
that token $t_u$ visits agent $v$ within $mnd$ steps in expectation.
Therefore, by Markov's inequality, 
$t_u$ visits $v$ within $2mnd$ steps
with probability at least $1/2$.
Therefore, they meet within $2\log_2 n \cdot (2mnd)=4mnd \log_2 n$ steps with probability at least $1-1/n^2$.
By the union bound,
all the $n$ tokens $(t_u)_{u \in V_G}$ visit $v$ within $4mnd \log_2 n$ steps
with probability $1-O(1/n)$,
from which the lemmas immediately follows.
\end{proof}

\begin{lemma}[\cite{Sud+20B}]
\label{lem:meet_each_other}
In execution $\Xi$,
all the $n$ tokens meet each other
within $O(mn^2 d \log n)$ steps in expectation.
\end{lemma}
\begin{proof}
Let $u$ and $v$ be any two distinct agents.
It suffices to show that two tokens $t_u$ and $t_v$
meet within $O(mn^2 d \log n)$ steps with probability $1-O(1/n^3)$:
then, by the union bound, all tokens meet each other
with probability $1-O(1/n)$ in every $O(mn^2 d \log n)$ steps,
yielding the lemma.

Let 
$\states = \{(w_1,w_2,s) \mid w_1, w_2 \in G, w_1 \neq w_2, s \in \{0,1\}\}$. 
Consider a Markov chain $\chainY = \{Y(t) \mid t=0,1,\dots\}$,
where
each $Y(t)$ belongs to $\states$.
For each state $(w_1,w_2,s) \in \states$,
$w_1$ and $w_2$ represent the locations of $t_u$ and $t_v$,
respectively, while the last element $s \in \{0,1\}$
is called the \emph{flag} of the state and 
will be used for the simplicity of analysis.
For $t \ge 1$ and $x,y \in \states$,
the probability
$\Pr(Y(t) = y \mid Y(t-1) = x)$ is independent of $t$,
denoted by $P_{\chainY}(x,y)$.
For $(a,b,s_1), (c,d,s_2) \in \states$,
we write $(a,b,s_1) \to (c,d,s_2)$
if (i) $(a,c) \in E_G \land b=d \land s_1 = s_2$,
(ii) $a=c \land (b,d) \in E_G \land s_1 = s_2$,
or (iii) $(a,b) \in E_G \land a=d \land b=c \land s_1 \neq s_2$.
Intuitively, the first (resp.~second) case represents that
$t_u$ (resp.~$t_v$) moves from $a$ to $c$
(resp.~from $b$ to $d$).
The third case represents that
the agents $a$ and $b$ swaps the tokens $t_u$ and $t_v$.
Note that we keep the flag of the state same
in the first and second case,
while we flip the flag, from 0 to 1 or from 1 to 0,
in the third case.
Thus, we can bound the expected number of steps before $t_u$ and $t_v$ meet in execution $\Xi$
by bounding the expected number of transition steps
before the flag is flipped in this Markov chain $\chainY$
that starts from $Y(0)=(u,v,0)$.
Probability $P_{\chainY}(x,y)$ is defined as follows:
$P_{\chainY}(x,y)=1/m$ if $x \to y$,
$P_{\chainY}(x,y)=1-|\{z \mid x \to z\}|/m$
if $x=y$, $P_{\chainY}(x,y)= 0$ otherwise.
The symmetric structure of the chain, \ie $P_{\chainY}(x,y)=P_{\chainY}(y,x)$
for all $x,y \in \states$,
gives $\sum_{x \in \states} P_{\chainY}(x,y)=1$
for any $y \in \states$.
Therefore, $\piY=(\piX(y_1),\piY(y_2),\dots,\piY(y_{2n(n-1)}))=\left(\frac{1}{2n(n-1)},\frac{1}{2n(n-1)},\dots,\frac{1}{2n(n-1)}\right)$
is the unique stationary distribution on $\chainY$
(\ie $\piY P_{\chainY} = \piY$), where $\states=\{y_1,y_2,\dots,y_{2n(n-1)}\}$.
For $x,y \in \states$,
we define \emph{the hitting time} $H_{\chainY}(x,y)$
as the expected number of transition steps in the chain
$\chainY$ from state $x$ to $y$.
We have $H_{\chainY}(z,z) = 1/\pi(z) = 2n(n-1)$
for any $z \in S$.
We also have
$H_{\chainY}(z,z)=1 + \sum_{w \in S \suchthat w \to z} (1/m) \cdot H_{\chainY}(w,z)$.
Therefore,
$\sum_{w \in S \suchthat w \to z} H_{\chainY}(w,z) < 2mn(n-1)$.
Thus, for any two distinct $x,y \in S$ such that $x \to y$, 
we have $H_{\chainY}(x,y) < 2mn^2$.

Let $w_0,w_1,\dots,w_l$ ($w_0=u$, $w_l=v$) 
be an arbitrary shortest path from $u$ to $v$ in $G$.
Clearly, the expected number of steps until two tokens
$t_u$ and $t_v$ meet in execution $\Xi$, say $M_{u,v}$,
is upper bounded
by $\sum_{i=0}^{l-2} H_{\chainY}((w_i,v,0),(w_{i+1},v,0))+H_{\chainY}((w_{l-1},v,0),(v,w_{l-1},1))$.
Since $(w_i,v,0) \to (w_{i+1},v,0)$ holds
for any $i=0,1,\dots,l-2$
and $(w_{l-1},v,0) \to (v,w_{l-1},1)$ holds,
we have $M_{u,v} < 2mn^2 l \le 2mn^2 d$.
Therefore, by Markov's inequality, 
$t_u$ and $t_v$ meet within $4mn^2 d$ steps
with probability at least $1/2$.
Hence, they meet within $3\lceil \log_2 n \rceil \cdot (4mn^2 d)=12mn^2d \lceil \log_2 n \rceil$ steps with probability at least $1-1/n^3$.
\end{proof}

\begin{lemma}
\label{lem:some_drift} 
Let $k$ be any positive integer.
There exists some agent $v \in V_G$ such that 
the expected number of steps until token $t_v$ moves 
$k$ times is $O(nk)$.
\end{lemma}
\begin{proof}
Let $s_u$ be the expected number of steps until
token $t_u$ moves $k$ times.
It suffices to show $\min_{u \in V_G} s_u = O(nk)$.
To analyze $\min_{u \in V_G} s_u$,
we introduce a Markov chain
$\chainZ = \{Z(t) \mid t=0,1,\dots\}$,
where each $Z(t) \in V_G$
represents the location of a token
 (\ie the agents that the token stays on)
after it moves $t$ times.
For $t \ge 1$ and $x,y \in V_G$,
the probability
$\Pr(Z(t) = y \mid Z(t-1) = x)$ is independent of $t$,
denoted by $P_{\chainZ}(x,y)$.
Probability $P_{\chainZ}(x,y)$ is calculated as follows:
$P_{\chainZ}(x,y)=1/\delta_x$ if $(x,y) \in E_G$,
$P_{\chainZ}(x,y)= 0$ otherwise,
where $\delta_x = |N_G(v)|$.
Let $V_G=\{z_1,z_2,\dots,z_n\}$.
Then,
$\piZ=(\piZ(z_i))_{i=1,2,\dots,n}$ is a stationary distribution\footnotemark{},
where $\piZ(z_i)= \delta_{z_i}/2m$
because $\piZ P_{\chainZ} = \piZ$.
\footnotetext{
This Markov chain $\chainZ$ is not ergodic when $G$ is bipartite.
However, this does not matter in this proof because
we do not use the recurrent time $H_{\chainZ}(z,z)$ unlike the proofs of Lemmas \ref{lem:hittingtime} and \ref{lem:meet_each_other}.
}


If a token visits agent $w \in V_G$,
then it needs $m/\delta_w$ steps in expectation
to leave $w$, \ie move to another agent from $w$.
Thus, we assign each agent $x \in V_G$
its \emph{weight} $W(x) \overset{\text{def}}{=} m/\delta_x$;
then, $s_u = \ex{\sum_{t=0}^{k-1} W(Z(t)) \longmid Z(0) = u}$
holds for any $u \in V_G$.
Assume that the initial state $Z_0$ is now set according to
the stationary distribution,
\ie $\Pr(Z(0) = z_i) = \piZ(z_i) = \delta_{z_i}/2m$ for any $i=1,2,\dots, n$.
Since $\piZ$ is a stationary distribution,
we always have the same distribution thereafter,
that is, we have  $\Pr(Z(t) = z_i) = \piZ(z_i)$
for any $t = 0,1,\dots$ and $i=0,1,\dots, n$.
Therefore, under this assumption, 
we have
\begin{align*}
\ex{\sum_{t=0}^{k-1} W(Z(t))}
= k \sum_{i=0}^n \piZ(z_i) W(z_i)
= k \sum_{i=0}^n \frac{\delta_{z_i}}{2m} \cdot \frac{m}{\delta_{z_i}}=\frac{kn}{2}. 
\end{align*}
We also have
$\ex{\sum_{t=0}^{k-1} W(Z(t))} = \sum_{u \in V_G} \piZ(u) s_u$.
Since $\sum_{u \in V_G} \piZ(u)=1$,
there must be at least one agent $u \in V_G$ such that
$s_u \le kn/2$.
Thus, $\min_{u \in V_G} s_u = O(kn)$.
\end{proof}

\begin{lemma}
\label{lem:drift} 
Let $k$ be any positive integer.
For any $v \in V_G$,
the expected number of steps until token $t_v$ moves 
$k$ times is $O(nk+mnd)$.
\end{lemma}
\begin{proof}
By Lemma \ref{lem:some_drift},
there exists an agent $u \in V_G$ such that
after visiting $u$, token $t_v$ moves $k$ times
within $O(nk)$ steps in expectation. 
By Lemma \ref{lem:hittingtime}, 
token $t_v$ visits $u$ within $mnd$ steps in expectation.
In total, token $t_v$ moves $k$ times
within $O(nk+mnd)$ steps in expectation.
\end{proof}

\section{Leader Election and Ranking}
\label{sec:ranking}

%
The goal of this section is to give a necessary and sufficient
condition to solve $\ranking$ and $\elect$ on
knowledge $\nu$, provided that $\mu$ gives no information,
\ie $\mu= \mdomain$.
For a necessary condition, we have the following lemma.

\begin{lemma}[\cite{jikoantei,SIW12,Sud+20}]
\label{lem:rank_impossible} 
Given knowledge $\nu$ and $\mu$,
there exists no self-stabilizing protocol that solves $\elect$
in arbitrary graphs
if $\graph_{n_1,*} \cup \graph_{n_2,*} \subseteq \graph_{\nu,\mu}$for some two distinct $n_1,n_2 \in \ndomain$.
\end{lemma}
\begin{proof}
The lemma immediately follows from the fact that
there exists no self-stabilizing protocol
that solves $\elect$ in complete graphs of two different sizes,
\ie both in $K_{n_1}$ and $K_{n_2}$
for any two integers $n_1 > n_2 \ge 2$.
As mentioned in Section \ref{sec:intro},
Sudo \etal~\cite{Sud+20} gave how to prove this fact
based on the proofs of \cite{jikoantei,SIW12}.
\end{proof}

To give a sufficient condition,
we give a self-stabilizing protocol
 $\prank$, which solves the ranking problem ($\ranking$)
in arbitrary graphs given the knowledge of the exact number of agents in a population.
Specifically, this protocol assumes that the given knowledge $\nu$
satisfies $|\nu|=1$ while it does not care about the number of interactable pairs,
that is, $\prank(\nu,\mu)$ works even if $\mu$ does not give any knowledge (\ie $\mu=\mdomain$).
Let $n$ be the integer such that $\nu = \{n\}$.

\begin{algorithm}[t]
\caption{$\prank(\nu,\mu)$}
\label{al:prank}
\textbf{Assumption}:
$|\nu|=1$.
(Let $\nu=\{n\}$.)
\vblank


\textbf{Variables}:
\begin{algorithmic}
\STATE $\idA, \idT \in \{0,1,\dots,n-1\}$
\STATE $\colorA \in \{\shiro,\aka,\ao\}$, $\colorT \in \{\aka,\ao\}$, $\timerT \in \{0,1,\dots,\tmax\}$
\end{algorithmic}
\vblank

\textbf{Output function} $\outputs$: $\idA$
\vblank

\textbf{Interaction} between initiator $a_0$ and responder $a_1$:
\hspace{-1cm} \verb| |
\begin{algorithmic}[1]  
 \setcounter{ALC@line}{0}
 \STATE $(a_0.\idT,a_0.\colorT,a_0.\timerT) \swap (a_1.\idT,a_1.\colorT,a_1.\timerT)$ \\
 \COMMENT{Execute the random walk of two tokens}
 \STATE \ifthen{
$a_0.\idT = a_1.\idT$
}{
$a_1.\idT \gets a_1.\idT + 1 \pmod{n}$
}
 \STATE \foralldo{
$i \in \{0,1\}$
}{
$a_i.\timerT \gets \max(0,a_i.\timerT-1)$
}
\vblank
 \FORALL{$i \in \{0,1\}$ such that $a_i.\idA = a_i.\idT$}
 \STATE \ifthen{$a_i.\colorA=\shiro$}{$a_i.\colorA \gets a_i.\colorT$}
 \IF{$a_i.\colorA \neq a_i.\colorT$}
 \STATE $a_i.\idA \gets a_i.\idA+1 \pmod{n}$
 \STATE $a_i.\colorA \gets \shiro$
 \ELSIF{$a_i.\timerT=0$}
 \STATE $a_i.\timerT \gets \tmax$
 \STATE
 \ifthen{
 $a_i.\colorA=\aka$
 }{
 $a_i.\colorA \gets a_i.\colorT \gets \ao$
 }
 \STATE
 \ifthen{
 $a_i.\colorA=\ao$
 }{
 $a_i.\colorA \gets a_i.\colorT \gets \aka$
 }
 \ENDIF
 \ENDFOR
\end{algorithmic}
\end{algorithm}

If we focus only on complete graphs,
the following simple algorithm \cite{SIW12}
is enough to solve self-stabilizing ranking
with the exact knowledge $n$ of agents:
\begin{itemize}
 \item Each agent $v$ has only one variable $v.\id \in \{0,1,\dots,n-1\}$, and
 \item Every time two agents with the same $\id$ meet, 
one of them (the initiator) increases its $\id$ by one modulo $n$.
\end{itemize}
%
Since this algorithm assumes complete graphs,
every pair of agents in the population eventually has interactions. Therefore, as long as two agents have the same identifiers, they eventually meet and the collision of their identifiers is resolved. However, this algorithm does not work in arbitrary graphs, even if the exact number of agents is given.
This is because some pair of agents may not be interactable
in an arbitrary graph, then they cannot resolve the conflicts of their identifiers by meeting each other.

Protocol $\prank$ detects the conflicts
between any (possibly non-interactable)
two agents by traversing $n$ tokens in a population where
each agent always has exactly one token.
This protocol is inspired by a self-stabilizing leader election protocol
with \emph{oracles} given by Beauquier \etal~\cite{twooracles},
where the agents traverse exactly one token in a population.

The pseudocode of $\prank$ is shown in Algorithm \ref{al:prank}.
Our goal is to assign the agents the distinct labels $0,1,\dots,n-1$.
Each agent $v$ stores its label in a variable $v.\idA \in \{0,1,\dots,n-1\}$
and outputs it as it is.
To detect and resolve the conflicts of the labels
in arbitrary graphs, 
each agent
maintains four other variables $\idT \in \{0,1,\dots,n-1\}$,
$\colorA \in \{\shiro,\aka,\ao\}$,
$\colorT \in \{\aka, \ao\}$,
and $\timerT \in \{0,1,\dots,\tmax\}$,
where $\tmax$ is a sufficiently large $\Omega(mn)$ value
and $m$ is the number of interactable pairs in the population.
We will explain later how to assign $\tmax$ such a value.
We say that $v$ has a token labeled $x$ if $v.\idT=x$.
Each agent $v$ has one color, 
white ($\shiro$), red  ($\aka$), or blue ($\ao$),
while $v$'s token has one color,
red ($\aka$) or blue ($\ao$),
maintained by variables $v.\idA$ and $v.\idT$, respectively.

The tokens always make \emph{the random walk}:
two agents swap their tokens whenever two agents interact
(Line 1). 
If the two tokens have the same label,
one of them increments its label modulo $n$
(Line 2).
Since all tokens meet each other infinitely often
by the random walk, 
they eventually have mutually distinct labels ($\idT$),
after which they never change their labels.
Thereafter, the conflicts of labels among the agents
are resolved by using the tokens. 
Let $x$ be any integer in $\{0,1,\dots,n-1\}$
and denote the token labeled $x$ by $T_x$.
Ideally, an agent labeled $x$
always has the same color as that of $T_x$.
Consider the case that an agent labeled $x$, say $v$, meets $T_x$,
and $v$ and $T_x$ have different colors, blue and red.
Then, $v$ suspects that there is another agent labeled $x$,
and $v$ increases its label by one modulo $n$
(Line 7).
The agent $v$, now labeled $x+1 \pmod{n}$,
changes its color to white (Line 8).
When $v$ meets $T_{x+1 \pmod{n}}$ the next time,
it copies the color of the token to its color
to synchronize a color with $T_{x+1 \pmod{n}}$.
Token $T_x$ changes its color periodically.
Specifically, $T_x$ decreases its $\timerT$
whenever it moves unless $\timerT$ already reaches zero
(Line 3).
If token $T_x$ meets an agent labeled $x$,
they have the same color, and the timer of the token is zero,
then they change their color from blue to red or from red to blue
(Lines 11--12).
If there are two or more agents labeled $x$,
this multiplicity is eventually detected
because $T_x$ makes a random walk forever:
$T_x$ eventually meets an agent labeled $x$ with a different color.
By repeating this procedure,
the population eventually reaches a configuration
where all the agents have distinct labels
and the agent labeled $x$ has the same color as that of $T_x$
for all $x = 0,1,\dots,n-1$.
No agent changes its label thereafter.

Note that this protocol works even if we do not use variable $\timerT$
and color $\shiro$.
We introduce them to make this protocol faster
under the uniformly random scheduler.
In the rest of this section, we prove the following theorem.

\begin{theorem}
\label{th:prank}
Given knowledge $\nu$ and $\mu$,
$\prank(\nu,\mu)$ is a self-stabilizing protocol
that solves $\ranking$ in arbitrary graphs
if $\nu = \{n\}$ for some integer $n$,
regardless of $\mu$.
Starting from any configuration $C_0$
on any population $G=(V_G,E_G) \in \graph_{n,*}$,
the execution of $\prank(\nu,\mu)$
under the uniformly random scheduler (\ie $\Xi_{\prank(\nu,\mu)}(G,C_0,\rs)$) reaches a safe configuration
within $O(mn^3d \log n + n^2 \tmax)$ steps in expectation,
where $m=|E_G|/2$ and $d$ is the diameter of $G$.
Each agent uses $O(\log n)$ bits of memory space
to execute $\prank(\nu,\mu)$.
\end{theorem}

Recall that we require parameter $\tmax$ to be a sufficiently large
$\Omega(mn)$ value.
If an upper bound $M$ of $m$ such that $M=\Theta(m)$ is obtained from knowledge $\mu$,
we can substitute a sufficiently large $\Theta(mn)$ value for $\tmax$.
Then, $\prank(\nu,\mu)$ converges in $O(mn^3d \log n)$ steps in expectation.
Even if such $M$ is not obtained from $\mu$, \eg $\mu = \mdomain$,
we can substitute a sufficiently large $\Theta(n^3)$ value for $\tmax$.
Then, $\prank(\nu,\mu)$ converges in $O(mn^3d \log n+ n^5)$ steps in expectation.

In the rest of this section,
we fix a population $G=(V_G,E_G) \in \graph_{n,*}$,
let $m=|E_G|$/2, and let $d$ be the diameter of $G$.
To prove Theorem \ref{th:prank},
we define three sets $\stoken$, $\ssync$, and $\srank$
of configurations in $\call(\prank(\nu,\mu),G)$ as follows.

\begin{itemize}
 \item $\stoken$:
the set of all the configurations
in $\call(\prank(\nu,\mu),G)$
where all tokens have distinct labels,
\ie $\forall u,v \in V_G:u.\idT \neq v.\idT$.
In a configuration in $\stoken$,
there exists exactly one token labeled $x$ in the population
for each $x \in \{0,1,\dots,n-1\}$.
We use notation $T_x$ both to denote the unique token labeled by $x$
and to denote the agent on which
this token \emph{currently} stays.
 \item $\ssync$: the set of all the configurations in $\stoken$
where 
proposition  $\qtoken(x) \overset{\text{def}}{\equiv} V_G(x) \neq \emptyset \Rightarrow (\exists u \in V_G(x): u.\colorA = T_x.\colorT \lor u.\colorA = \shiro)$
holds for any $x \in \{0,1,\dots,n-1\}$,
where $V_G(x) \overset{\text{def}}{=} \{v \in V \mid v.\idA=x\}$.

 \item $\srank$: the set of all the configurations in $\ssync$
where all the agents in $V_G$ have distinct labels,
that is, $\forall u,v \in V_G: u.\idA \neq v.\idA$.
\end{itemize}

\begin{lemma}
\label{lem:stoken_closed}
The set $\stoken$ is closed for $\prank(\nu,\mu)$.
\end{lemma}
\begin{proof}
A token changes its label only if it meets another token
with the same label. Hence, no token changes its label
in an execution starting from a configuration in $\stoken$.
\end{proof}

\begin{lemma}
\label{lem:qtoken}
Let $x \in \{0,1,\dots,n-1\}$.
In an execution of $\prank(\nu,\mu)$
starting from a configuration in $\stoken$, once $\qtoken(x)$ holds, it always holds thereafter.
\end{lemma}
\begin{proof}
This lemma holds because (i) an agent must be white just after
it changes its label from $x-1 \pmod{n}$ to $x$,
(ii)
a white agent labeled $x$ changes its color only when
token $T_x$ visits it at an interaction,
at which this white agent gets the same color as that of $T_x$,
(iii)
an agent labeled $x$ with the same color as that of $T_x$
changes its color only when token $T_x$ visits it at an interaction,
at which this agent and $T_x$ get the same new color.
\end{proof}

\begin{lemma}
\label{lem:ssync_closed}
The set $\ssync$ is closed for $\prank(\nu,\mu)$.
\end{lemma}
\begin{proof}
The lemma immediately follows from Lemma \ref{lem:qtoken}.
\end{proof}

\begin{lemma}
\label{lem:never_zero}
Let $x \in \{0,1,\dots,n-1\}$.
In an execution of $\prank(\nu,\mu)$
starting from a configuration in $\ssync$,
once at least one agent is labeled $x$, 
the number of agents labeled $x$
never becomes zero thereafter.
\end{lemma}
\begin{proof}
This lemma holds in the same way
as the proof of Lemma \ref{lem:qtoken}.
\end{proof}

\begin{lemma}
\label{lem:srank_closed}
The set $\srank$ is closed for $\prank(\nu,\mu)$.
\end{lemma}
\begin{proof}
The lemma immediately follows from
Lemmas \ref{lem:ssync_closed} and \ref{lem:never_zero}.
\end{proof}

The following lemma is useful to analyze
the expected number of steps
required to reach a configuration in $\srank$
in an execution of $\prank(\nu,\mu)$.
\begin{lemma}
\label{lem:game} 
Consider the following game with $n$ players
$p_0, p_1, \dots, p_{n-1}$.
Each player always has one state in $\{0,1,\dots,n-1\}$.
At each step, an arbitrary pair of players is selected 
and they check the states of each other.
If they have the same state, one of them increases its state
by one modulo $n$. Otherwise, they do not change their states.
Starting this game from any configuration (\ie any combination of the states of all players),
there is at least one state $z \in \{0,1,\dots,n-1\}$ such that 
no player changes its state from $z-1 \pmod{n}$ to $z$.
The set of such states is uniquely determined by a configuration
from which the game starts.
\end{lemma}

\begin{proof}
Fix an initial configuration $\psi_0=(k_0,k_1,\dots,k_{n-1})$,
where $k_i$ represents the number of agents in state $i$
in the configuration.
In this proof, we make every addition and subtraction
in modulo $n$ and omit the notation ``$\pmod{n}$''.
It is trivial that for any $x \in \{0,1,\dots,n-1\}$,
no player changes its state from $x-1$
to $x$ 
if and only if $x$ satisfies 
$\sum_{j=1}^{i} k_{x-j} \le i$
for all $i \in \{1,2,\dots,n-1\}$.
Therefore, the set of states $z$ such that 
no player changes its state from $z-1$ to $z$
is uniquely determined by the initial configuration $\psi_0$.

By the uniqueness of the above set,
it suffices to show that
for any execution $\Xi$ of this game starting from $\psi_0$,
there is a state $z \in \{0,1,\dots,n-1\}$
such that no player changes its state from $z-1$ to $z$
in $\Xi$.
We say that a state $x \in \{0,1,\dots,n-1\}$ is \emph{filled}
if at least one player is in state $x$.
By definition of this game, once $x$ is filled, 
$x$ is always filled thereafter. 
If there is a state $z$ that is never filled in $\Xi$,
no player changes its state from $z-1$ to $z$.
Suppose the other case and
let $z$ be the state that is filled for the last time in execution $\Xi$.
By definition,
when $z$ gets filled, all the $n$ states are filled, 
which yields that all
the $n$ players have mutually distinct states at this time.
Therefore, no player never changes its state from $z-1$ to $z$
in execution $\Xi$.


\end{proof}


\begin{lemma}
\label{lem:stoken}
Starting from any configuration $C_0 \in \call(\prank(\nu,\mu),G)$,
an execution of $\prank(\nu,\mu)$ under the uniformly
random scheduler (\ie $\Xi_{P(\nu,\mu)}(G,C_0,\rs)$)
reaches a configuration
in $\stoken$ within $O(mn^3 d \log n)$ steps in expectation.
\end{lemma}
\begin{proof}
By Lemma \ref{lem:game},
there exists an integer $z \in \{0,1,\dots,n-1\}$
such that no \emph{token} changes its label
from $z-1 \pmod{n}$ to $z$.
Then, the number of tokens labeled $z$ becomes exactly one
before or when all the tokens meet each other.
Since Sudo \etal~\cite{Sud+20B} proved that
$n$ tokens making random walks in arbitrary graphs
meet each other within $O(mn^2 d \log n)$ steps in expectation,
the number of tokens labeled $z$ becomes exactly one
within  $O(mn^2 d \log n)$ steps in expectation.
Thereafter, no token changes its label
from $z$ to $z+1 \pmod{n}$. 
Hence, the number of tokens labeled $z+1 \pmod{n}$
becomes one in the next $O(mn^2 d \log n)$ steps 
in the same way.
Repeating this procedure,
all the tokens have distinct labels
within $O(mn^3 d \log n)$ steps in expectation.
\end{proof}

\begin{lemma}
\label{lem:ssync}
Starting from any configuration $C_0 \in \stoken$,
an execution of $\prank(\nu,\mu)$ under the uniformly
random scheduler (\ie $\Xi_{P(\nu,\mu)}(G,C_0,\rs)$)
reaches a configuration
in $\ssync$ within $O(mn^3)$ steps in expectation.
\end{lemma}

\begin{proof}
By Lemmas \ref{lem:stoken_closed} and \ref{lem:qtoken},
it suffices to show that for each $x \in \{0,1,\dots,n-1\}$,
$\qtoken(x)$ becomes true
within $O(mn^2)$ steps in expectation in an execution
of $\prank(\nu,\mu)$ starting from $C_0$.
We have $\qtoken(x) = \fl$ if and only if
there exists at least one agent labeled $x$
and all of them have colors different from that of $T_x$ (\ie the token labeled $x$).
Even if $\qtoken(x) = \fl$ in $C_0$,
$\qtoken(x)$ becomes true before or when $T_x$ meets all of them.
By Lemma \ref{lem:covertime},
$T_x$ visits (\ie meets) all agents
within $O(mn^2)$ steps in expectation,
from which the lemma follows.
\end{proof}

\begin{lemma}
\label{lem:srank}
Assume that $\tmax$ is sufficiently
large $\Omega(mn)$ value.
Starting from any configuration $C_0 \in \ssync$,
an execution of $\prank(\nu,\mu)$ under the uniformly
random scheduler (\ie $\Xi_{P(\nu,\mu)}(G,C_0,\rs)$)
reaches a configuration
in $\srank$ within $O(mn^3 + n^2 \tmax)$ steps in expectation.
\end{lemma}

\begin{proof}
By Lemmas \ref{lem:never_zero} and \ref{lem:game},
there exists an integer $z \in \{0,1,\dots,n-1\}$
such that no \emph{agent} changes its label
from $z-1 \pmod{n}$ to $z$.
Therefore, at least one agent is labeled $z$ in $C_0$.
All of them get non-white color, \ie blue or red,
or get a new label $z+1 \pmod{n}$
before or when $T_z$ meets all agents,
which requires only $O(mn^2)$ steps in expectation
(See Lemma \ref{lem:covertime}).
Without loss of generality, 
we assume that token $T_z$ is red at this time.
By Lemma \ref{lem:ssync_closed}, there is at least one
\emph{red} agent labeled $z$.
After that, by Lemma \ref{lem:drift}
the $\timerT$ of $T_z$ becomes zero
within $O(n\tmax)$ steps in expectation.
In the next $O(mn^2)$ steps in expectation,  
$T_z$ meets a red agent labeled $z$,
at which $T_z$ and this agent changes their colors
to \emph{blue}, and $T_z$ resets its $\timerT$ to $\tmax$.
It is well known that
a token making the random walk visits all nodes of
any undirected graph within $O(mn)$ moves in expectation.
Since a token decreases its $\timerT$ only by one
every time it moves,
$T_z$ meets all agents and makes each agent labeled $z$ blue or pushes it to the next label (\ie $z+1 \pmod{n}$) before its $\timerT$ reaches zero again from $\tmax=\Omega(mn)$,
with probability $1-p$ for any small constant $p$,
by Markov's inequality.
By Lemma \ref{lem:covertime},
this requires only $O(mn^2)$ steps in expectation.
Similarly, (i) the $\timerT$ of $T_z$ becomes zero again in the next
$O(n \tmax)$ steps,
(ii) $T_x$ meets a blue agent labeled $z$, say $v$,
in the next $O(mn^2)$ steps, at which 
$T_x$ and $v$ become red,
and (iii) $T_x$ meets all agents and
pushes all agents labeled $z$ except for $v$ to the next label
in the next $O(mn^2)$ steps in expectation and
with probability $1-p$ for any small constant $p$.
Therefore, the number of agents labeled $z$ becomes one 
within $O(mn^2 + n \tmax)$ steps in expectation.
After that, no agent changes its label from $z$ to $z+1 \pmod{n}$.
Hence, the number of agents labeled $z+1 \pmod{n}$ becomes one 
in the next $O(mn^2 + n\tmax)$ steps in expectation
by the same reason.
Repeating this procedure,
all agents get mutually distinct labels (\ie $\idA$)
within $O(mn^3 + n^2 \tmax)$ steps in expectation.
\end{proof}

\begin{proof}[of Theorem \ref{th:prank}]
By Lemmas \ref{lem:stoken}, \ref{lem:ssync}, and \ref{lem:srank},
$\Xi_{\prank(\nu,\mu)}(G,C_0,\rs)$ reaches a configuration
in $\srank$
within $O(mn^3d \log n + n^2 \tmax)$ steps in expectation.
By Lemma \ref{lem:srank_closed},
every configuration in $\srank$ is a safe configuration
for the ranking problem.
\end{proof}

\begin{theorem}
\label{th:ranking}
Let $\nu$ be any subset of $\ndomain$
and let $\mu = \mdomain$.
Given knowledge $\nu$ and $\mu\ (=\mdomain)$,
there exists a self-stabilizing protocol 
that solves $\elect$ and $\ranking$ in arbitrary graphs
if and only if
the agents know the exact number of agents
\ie $\graph_{\nu,\mu}=\graph_{n,*}$
for some $n \in \ndomain$.     
\end{theorem}
\begin{proof}
The theorem immediately follows from Lemma \ref{lem:rank_impossible}, Theorem \ref{th:prank},
and the fact that $\elect \preceq \ranking$.
\end{proof}

\section{Degree Recognition and Neighbor Recognition}
\label{sec:neighbor}
Our goal is to prove the negative and positive propositions 
for $\degree$ and $\neighbor$
introduced in Section \ref{sec:intro}.
First, we prove the negative proposition.
\begin{lemma}
\label{lem:degree_impossible} 
Let $\nu$ and $\mu$ be any sets such that
$\nu \subseteq \ndomain$ and $\mu \subseteq \mdomain$.
There exists no self-stabilizing protocol
that solves $\degree$ in all graphs in $\graph_{\nu,\mu}$
 if $\graph_{n,m_1} \cup \graph_{n,m_2} \subseteq \graph_{\nu,\mu}$ holds for some $n \in \ndomain$ and some
distinct $m_1,m_2 \in \mdomain$ such that
$\graph_{n,m_1} \neq \emptyset$ and $\graph_{n,m_2} \neq \emptyset$.
\end{lemma}
\begin{proof}
Assume $m_1 < m_2$ without loss of generality.
By definition, there must exist two graphs
$G'=(V_{G'},E_{G'}) \in \graph_{n,m_1}$
and $G''=(V_{G''},E_{G''}) \in \graph_{n,m_2}$
such that 
$V_{G'}=V_{G''}$ and $E_{G'} \subset E_{G''}$.
Then, there exists at least one agent $v \in V_{G''}$
such that its degree differs in $G'$ and $G''$.
Let $\delta'$ and $\delta''$ be the degrees of $v$
in $G'$ and $G''$, respectively.
Assume for contradiction that there is a self-stabilizing protocol
$P(\nu,\mu)$ that solves $\degree$ both in $G'$ and $G''$.
By definition, there must be at least one safe configuration
$S$ of protocol $P(\nu,\mu)$ on $G''$ for $\degree$.
In every execution of $P(\nu,\mu)$ starting from $S$ on $G''$,
agent $v$ must always output $\delta''$ as its degree.
The configuration $S$ can also be a configuration on $G'$
because $V_{G'}=V_{G''}$.
Since $P(\nu,\mu)$ is self-stabilizing in $G'$,
there must be a finite sequence of interactions
$\gamma_0,\gamma_1,\dots,\gamma_t$ of $G'$
that put configuration $S$ to a configuration where
$v$ outputs $\delta'$ as its degree.
Since $E_{G'} \subset E_{G''}$,
$\gamma_0,\gamma_1,\dots,\gamma_t$ is
also a sequence of interactions in $G''$.
This implies that this sequence changes the output of $v$
from $\delta''$ to $\delta'$
starting from a \emph{safe} configuration,
a contradiction. 
\end{proof}
%

\begin{algorithm}[t]
\caption{$\pneighbor(\nu,\mu)$}
\label{al:pneighbor}

\textbf{Assumption}:
$|\nu|=1$ and $|\mu|=1$.
(Let $\nu = \{n\}$ and $\mu = \{m\}$.)
\vblank


\textbf{Variables}:
\begin{algorithmic}
\STATE $\idA, \idT \in \{0,1,\dots,n-1\}$
\COMMENT{Updated only by $\prank$}
\STATE $\degreeT \in \{0,1,\dots,n\}$, $\dsum \in \{0,1,\dots,2m+1\}$
\STATE $\resetE \in \{0,1,\dots,\emax\}$, $\timerP \in \{0,1,\dots,\pmax\}$
\STATE $\neighbors, \counted \in 2^{\{0,1,\dots,n-1\}}$
\end{algorithmic}
\vblank

\textbf{Output function} $\outputs$: $(\idA,\neighbors)$
\vblank

\textbf{Interaction} between initiator $a_0$ and responder $a_1$:
\hspace{-1cm} \verb| |
\begin{algorithmic}[1]  
 \setcounter{ALC@line}{0}   
 \STATE Execute $\prank$ with substituting sufficiently large $\Theta(mn)$ value for $\tmax$.
 \STATE $a_0.\degreeT \swap a_1.\degreeT$\\
 \COMMENT{Execute the random walk of two tokens with $\prank$}
 \vblank

 \STATE $a_0.\resetE \gets a_1.\resetE \gets \max(0,a_0.\resetE-1,a_1.\resetE-1)$
 \STATE
 \ifthen{$a_0.\resetE > 0$}{$a_0.\neighbors \gets a_1.\neighbors \gets \emptyset$}
 \vblank

 \FORALL{$i\in\{0,1\}$}
 \STATE $a_i.\timerP \gets \max(0,a_i.\timerP-1)$
 \IF{$a_i.\timerP = 0$}
 \STATE $(a_i.\dsum,a_i.\counted,a_i.\timerP) \gets (0,\emptyset,\pmax)$
 \ENDIF
 \STATE $a_i.\neighbors \gets a_i.\neighbors \cup \{a_{1-i}.\idA\}$
 \STATE
 \ifthen{
 $a_i.\idA=a_i.\idT$
 }{
 $a_i.\degreeT \gets |a_i.\neighbors|$
 }
 \vblank 

 \IF{$a_i.\idT \notin a_i.\counted$}
 \STATE $a_i.\dsum \gets \min(2m+1,a_i.\dsum + a_i.\degreeT)$
 \STATE $a_i.\counted \gets a_i.\counted \cup \{a_i.\idT\}$
 \ENDIF
 \STATE \ifthen{$a_i.\dsum=2m+1$}{$a_i.\resetE \gets \emax$}
 \ENDFOR
\end{algorithmic}
%
%
%
\end{algorithm}

To prove the positive proposition,
we give a self-stabilizing protocol
$\pneighbor$, which solves the neighbor recognition problem
($\neighbor$) in arbitrary graphs
given the knowledge of the exact number of agents
and the exact number of interactable pairs,
that is, given knowledge $\nu$ and $\mu$ such that
$|\nu|=|\mu|=1$.
In the rest of this section,
let $n$ and $m$ be the integers such that 
$\nu = \{n\}$ and $\mu = \{m\}$.

The pseudocode of $\pneighbor$
is shown in Algorithm \ref{al:pneighbor}.
Our goal is to let the agents recognize
the set of their neighbors. 
Each agent $v$ stores its label in a variable $v.\idA \in \{0,1,\dots,n-1\}$ and the set of the labels assigned to its neighbors
in a variable $\neighbors \in 2^{\{0,1,\dots,n-1\}}$.
Each agent $v$ outputs $(v.\idA,v.\neighbors)$. 

We use $\prank$ as a sub-algorithm
to assign the agents the distinct labels $0,1,\dots,n-1$
and to let the $n$ tokens make the random walk.
Specifically, we first execute $\prank$ whenever two agents have an interaction (Line 1), substituting a sufficiently large $\Theta(mn)$ value for $\tmax$.
We do not update the variables used in $\prank$
in the other lines (Lines 2--17).
Therefore, by Theorem \ref{th:prank},
an execution of $\pneighbor$ starting from any configuration
reaches a configuration
in $\srank$ within $O(mn^2 d \log n)$ steps in expectation.
Hence, 
we need to consider
only an execution after reaching a configuration
in $\srank$. Then, we can assume that 
the population always has exactly one agent labeled $x$
and exactly one token labeled $x$ for each $x=\{0,1,\dots,n-1\}$.
We denote them by $A_x$ and $T_x$, respectively.

The agents compute their $\neighbors$ in a simple way:
every time two agents $u$ and $v$ have an interaction,
$u$ adds $v.\idA$ to $u.\neighbors$
and $v$ adds $u.\idA$ to $v.\neighbors$ (Line 10).
However, this simple way to compute $\neighbors$
is not enough to design a self-stabilizing protocol
because we consider an arbitrary initial configuration.
Specifically, in an initial configuration,
$v.\neighbors$ may include
$u.\idA$ for some $u \notin N_G(v)$.
We call such $u.\idA$ a \emph{fake label}.
To compute $v.\neighbors$ correctly, 
in addition to the above simple mechanism,
it suffices to detect the existence of a fake label
and reset the $\neighbors$ of all agents
to the empty set
if a fake label is detected.

Using the knowledge $\mu = \{m\}$,
we achieve the detection of fake labels
with the following strategy.
%
Each token $T_x$ carries $|A_x.\neighbors|$
in a variable $\degreeT \in \{0,1,\dots,n\}$ (Line 2).
Whenever $T_x$ meet $A_x$, the value of $T_x.\degreeT$
is updated by the current value of $|A_x.\neighbors|$
(Line 11).
Each agent always tries to estimate
$\sum_{v \in V_G} |v.\neighbors|$ using variables
$\dsum \in \{0,1,\dots,2m+1\}$,
$\counted \in 2^{\{0,1,\dots,n-1\}}$,
and $\timerP \in \{0,1,\dots,\pmax\}$,
where $\pmax$ is a sufficiently large $\Theta(mnd \log n)$ value.
It uses $\timerP$ as a count-down timer
to reset $\dsum$ and $\counted$ periodically.
Specifically, an agent $v$
decreases $v.\timerP$ by one every time 
it has an interaction and 
resets $v.\dsum$, $v.\counted$, and $v.\timerP$
to $0$, $\emptyset$, and $\pmax$, respectively,
when $v.\timerP$ reaches zero
(Lines 6-9).
Whenever agent $v$ meets $T_x$ such that $x \notin v.\counted$,
$v$ executes $v.\dsum \gets \min(2m+1,v.\dsum + T_x.\degreeT)$
and adds $x$ to $v.\counted$.
(Lines 12-15)
We expect $v.\dsum = \sum_{v \in V_G} \allowbreak |v.\neighbors|$
when $v$ meets all of $T_0, T_1, \dots, T_{n-1}$.
If $v.\dsum$ reaches $2m+1$,
agent $v$ concludes that 
at least one agent has a fake label,
\ie $u.\neighbors \not \subseteq \{w.\idA \mid w \in N_G(u)\}$
for some $u$.

When the existence of a fake label is detected,
we reset the $\neighbors$s of all agents 
using a variable $\resetE \in \{0,1,\dots,\emax\}$,
where $\emax$ is a sufficiently large $\Theta(n^2)$ value.
Specifically, when $v.\dsum=2m+1$ holds,
$v$ emits the error signal by setting variable
$v.\resetE$ to $\emax$
(Line 16).
Thereafter, the error signal is propagated
to the whole population via \emph{the larger value propagation}:
when two agents $u$ and $v$ meet,
they substitute
$\max(0,u.\resetE-1,v.\resetE-1)$ for their $\resetE$s.
(Line3).
Whenever an agent $v$ receives the error signal,
\ie $v.\resetE > 0$ holds,
it resets its $\neighbors$ to the empty set (Line 4).

Thus, even if some agent has fake labels at the beginning of an execution,
the population eventually reaches a configuration where no agent has fake labels
after the occurrence of the following events:
the existence of a fake label is detected,
the error signal propagates to the whole population,
and all agents reset their $\neighbors$s to the empty set.
Thereafter, for any $x \in \{0,1,\dots,n-1\}$,
$T_x$ eventually meets $A_x$,
after which $T_x.\degreeT \le |N_G(A_x)|$ always hold.
Hence, by the periodical reset of $\dsum$ and $\counted$,
the population eventually reach a configuration
from which no agent emits the error signal. 
Thereafter, the population will soon reach a configuration
that satisfies $v.\neighbors = \{u.\idA \mid u \in N_G(v)\}$
for all $v \in V_G$
by the above simple computation of $\neighbors$ (Line 10).
Once it reaches such a configuration, 
no agent changes its $\neighbors$.

\begin{theorem}
\label{th:pneighbor}
Given knowledge $\nu$ and $\mu$,
$\pneighbor(\nu,\mu)$ is a self-stabilizing protocol
that solves $\neighbor$ in arbitrary graphs
if $\nu = \{n\}$ and $\mu = \{m\}$
for some integers $n$ and $m$.
Starting from any configuration $C_0$
on any population $G=(V_G,E_G) \in \graph_{n,m}$,
the execution of $\pneighbor(\nu,\mu)$
under the uniformly random scheduler
(\ie $\Xi_{\pneighbor(\nu,\mu)}(G,C_0,\rs)$)
reaches a safe configuration
within $O(mn^3d \log n)$ steps in expectation,
where $m=|E_G|/2$ and $d$ is the diameter of $G$.
Each agent uses $O(n)$ bits of memory space
to execute $\pneighbor(\nu,\mu)$. 
\end{theorem}
\begin{proof}
Define $\nlabel(v) = \{u.\idA \mid u \in N_G(v)\}$
and define $\snofake$ as the set of all configurations
in $\srank$
where no agent has a fake label in its $\neighbors$,
that is, 
$v.\neighbors \subseteq \nlabel(v)$
 holds for all $v \in V_G$.

First, we show that
execution $\Xi = \Xi_{\pneighbor(\nu,\mu)}(G,C_0,\rs)$
reaches a configuration in $\snofake$
within $O(mn^3 d \log n)$ steps in expectation.
By Theorem \ref{th:prank},
$\Xi$ reaches a configuration $C'$ in $\srank$
within $O(mn^3 d \log n)$ steps in expectation
because $\tmax = \Theta(mn)$.
We assume $C' \notin \snofake$ because
otherwise we need not discuss anything.
Interactions happen between all interactable pairs
within $O(m \log n)$ steps in expectation.
Therefore, after reaching $C'$,
$\Xi$ reaches within $O(m \log n)$ steps in expectation
a configuration $C''$ where $\nlabel(v) \subseteq v.\neighbors$
for all $v \in V_G$
or a configuration where $u.\resetE > 0$ for some $u \in V_G$. In the former case, $\sum_{v \in V}|v.\neighbors| > 2m$ holds
in $C''$
since at least one agent has one or more fake labels
in its $\neighbors$. 
Thereafter, some agent $v$ decreases its $\timerP$
to zero and resets it to $\pmax$ in
the next $O(m \pmax)=O(m^2 n d \log n) \subseteq O(mn^3 d \log n)$ steps in expectation.
After that, 
$v$ meets all tokens within $O(mnd \log n)$ steps
in expectation.
(See \ref{lem:meet_all_tokens}.)
As a result, $v.\dsum$ reaches $2m+1$
and $v$ emits the error signal.
To conclude, after $\Xi$ reaches $C'$, 
some agent emits the error signal,
\ie it substitutes $\emax$ for its $\resetE$.
Since we set $\emax$ to a sufficiently large $\Theta(n^2)$ value,
the error signal is propagated to the whole population 
within $O(mn)$ steps with probability $1-O(1/n)$.
(See Lemma 5 in \cite{ninikanjiko}.)
Every time an agent receives the error signal,
it resets its $\neighbors$ to the empty set.
Therefore, $\Xi$ reaches a configuration in $\snofake$
within $O(mn^3d \log n)$ steps in expectation.

After entering $\snofake$, $\Xi$ reaches
within $O(mnd \log n)$ steps in expectation a configuration
where $\sum_{x =0,1,\dots,n-1}T_x.\degreeT \le 2m$ holds;
because every $T_x$ meets $A_x$ within $O(mnd)$ steps in expectation for every $x \in \{0,1,\dots,n-1\}$. Similarly, all agents reset their $\dsum$ and $\counted$
in the next $O(m \pmax)\subseteq O(mn^3d \log n)$ step in expectation. Thereafter, no agent sees $\dsum = 2m+1$, hence no agent emits the error signal, after which the error signal disappears from the population in the next $O(\emax \cdot m \log m) = O(mn^2 \log n)$
steps in expectation.
Therefore, interactions happen between all interactable pairs
in the next $O(m \log n)$ steps in expectation,
by which $v.\neighbors = \nlabel(v)$ holds for all $v \in V_G$.
After that, no agent $v$ changes $v.\neighbors$,
which yields that $\Xi$ has reached a safe configuration.

Each agent uses only $O(n)$ bits:
both variables $\neighbors$ and $\counted$
require $n$ bits 
and all other variables used in $\pneighbor$
require $O(\log n)$ bits.
\end{proof}

%




\section{Conclusion}
\label{sec:conclusion}
In this paper,
we clarified the solvability
of the leader election problem, the ranking problem,
the degree recognition problem, and
the neighbor recognition problem
by self-stabilizing population protocols
with knowledge of the number of nodes
and/or the number of edges in a network.
The protocols we gave in this paper
require \emph{exact} knowledge on
the number of agents and/or
the number of interactable pairs.
It is interesting and still open
whether \emph{ambiguous} knowledge
such as ``the number of interactable pairs is at most $M$''
and ``the number of agents is not a prime number''
is useful to design self-stabilizing population protocols.

%
%
%
\bibliographystyle{abbrv}
\bibliography{paper6}

\begin{thebibliography}{10}

\bibitem{AG15}
D.~Alistarh and R.~Gelashvili.
\newblock Polylogarithmic-time leader election in population protocols.
\newblock In {\em Proceedings of the 42nd International Colloquium on Automata,
  Languages, and Programming}, pages 479--491, 2015.

\bibitem{original}
D.~Angluin, J.~Aspnes, Z.~Diamadi, M.~J. Fischer, and R.~Peralta.
\newblock Computation in networks of passively mobile finite-state sensors.
\newblock {\em Distributed Computing}, 18(4):235--253, 2006.

\bibitem{fast}
D.~Angluin, J.~Aspnes, and D.~Eisenstat.
\newblock Fast computation by population protocols with a leader.
\newblock {\em Distributed Computing}, 21(3):183--199, 2008.

\bibitem{jikoantei}
D.~Angluin, J.~Aspnes, M.~J. Fischer, and H.~Jiang.
\newblock Self-stabilizing population protocols.
\newblock {\em ACM Transactions on Autonomous and Adaptive Systems}, 3(4):13,
  2008.

\bibitem{twooracles}
J.~Beauquier, P.~Blanchard, and J.~Burman.
\newblock Self-stabilizing leader election in population protocols over
  arbitrary communication graphs.
\newblock In {\em International Conference on Principles of Distributed
  Systems}, pages 38--52, 2013.

\bibitem{Bur+19}
J.~Burman, D.~Doty, T.~Nowak, E.~E. Severson, and C.~Xu.
\newblock Efficient self-stabilizing leader election in population protocols.
\newblock {\em arXiv preprint arXiv:1907.06068}, 2019.

\bibitem{SIW12}
S.~Cai, T.~Izumi, and K.~Wada.
\newblock How to prove impossibility under global fairness: On space complexity
  of self-stabilizing leader election on a population protocol model.
\newblock {\em Theory of Computing Systems}, 50(3):433--445, 2012.

\bibitem{onebit}
D.~Canepa and M.~G. Potop-Butucaru.
\newblock Stabilizing leader election in population protocols.
\newblock 2007.
\newblock http://hal.inria.fr/inria-00166632.

\bibitem{CC19}
H.-P. Chen and H.-L. Chen.
\newblock Self-stabilizing leader election.
\newblock In {\em Proceedings of the 2019 ACM Symposium on Principles of
  Distributed Computing}, pages 53--59, 2019.

\bibitem{CG17}
G.~Cordasco and L.~Gargano.
\newblock Space-optimal proportion consensus with population protocols.
\newblock In {\em International Symposium on Stabilization, Safety, and
  Security of Distributed Systems}, pages 384--398, 2017.

\bibitem{ganso}
E.~Dijkstra.
\newblock {Self-stabilizing systems in spite of distributed control}.
\newblock {\em Communications of the ACM}, 17(11):643--644, 1974.

\bibitem{oracle}
M.~J. Fischer and H.~Jiang.
\newblock Self-stabilizing leader election in networks of finite-state
  anonymous agents.
\newblock In {\em International Conference on Principles of Distributed
  Systems}, pages 395--409, 2006.

\bibitem{GSU18}
L.~G{\k{a}}sieniec, G.~Stachowiak, and P.~Uznanski.
\newblock Almost logarithmic-time space optimal leader election in population
  protocols.
\newblock In {\em The 31st ACM on Symposium on Parallelism in Algorithms and
  Architectures}, pages 93--102. ACM, 2019.

\bibitem{kakai}
T.~Izumi.
\newblock On space and time complexity of loosely-stabilizing leader election.
\newblock In {\em International Colloquium on Structural Information and
  Communication Complexity}, pages 299--312, 2015.

\bibitem{MNRS14}
G.~B. Mertzios, S.~E. Nikoletseas, C.~L. Raptopoulos, and P.~G. Spirakis.
\newblock Determining majority in networks with local interactions and very
  small local memory.
\newblock In {\em International Colloquium on Automata, Languages, and
  Programming}, pages 871--882, 2014.

\bibitem{Sud+20d}
Y.~Sudo, R.~Eguchi, T.~Izumi, and T.~Masuzawa.
\newblock Time-optimal loosely-stabilizing leader election in population
  protocols.
\newblock {\em arXiv preprint arXiv:2005.09944}, 2020.

\bibitem{samespeed}
Y.~Sudo, T.~Masuzawa, A.~K. Datta, and L.~L. Larmore.
\newblock The same speed timer in population protocols.
\newblock In {\em the 36th IEEE International Conference on Distributed
  Computing Systems}, pages 252--261, 2016.

\bibitem{kanjiko}
Y.~Sudo, J.~Nakamura, Y.~Yamauchi, F.~Ooshita, H.~Kakugawa, and T.~Masuzawa.
\newblock Loosely-stabilizing leader election in a population protocol model.
\newblock {\em Theoretical Computer Science}, 444:100--112, 2012.

\bibitem{SOI+19+}
Y.~Sudo, F.~Ooshita, T.~Izumi, H.~Kakugawa, and T.~Masuzawa.
\newblock Logarithmic expected-time leader election in population protocol
  model.
\newblock In {\em Proceedings of the 21st International Symposium on
  Stabilizing, Safety, and Security of Distributed Systems}, pages 323--337,
  2019.

\bibitem{ninikanjiko}
Y.~Sudo, F.~Ooshita, H.~Kakugawa, and T.~Masuzawa.
\newblock Loosely-stabilizing leader election on arbitrary graphs in population
  protocols.
\newblock In {\em International Conference on Principles of Distributed
  Systems}, pages 339--354, 2014.

\bibitem{Sud+20B}
Y.~Sudo, F.~Ooshita, H.~Kakugawa, and T.~Masuzawa.
\newblock Loosely stabilizing leader election on arbitrary graphs in population
  protocols without identifiers or random numbers.
\newblock {\em IEICE Transactions on Information and Systems}, 103(3):489--499,
  2020.

\bibitem{Sud+20}
Y.~Sudo, F.~Ooshita, H.~Kakugawa, T.~Masuzawa, A.~K. Datta, and L.~L. Larmore.
\newblock Loosely-stabilizing leader election with polylogarithmic convergence
  time.
\newblock {\em Theoretical Computer Science}, 806:617--631, 2020.

\end{thebibliography}
%
\end{document}